\documentclass{article}

\usepackage{fullpage}
\usepackage[small]{caption}
\usepackage[utf8]{inputenc} 
\usepackage[T1]{fontenc}    
\usepackage[hidelinks]{hyperref}      
\usepackage{url}            
\usepackage{booktabs}       
\usepackage{amsfonts}       
\usepackage{nicefrac}       
\usepackage{microtype}      
\usepackage{algpseudocode}
\usepackage{amsmath}
\usepackage{amssymb}
\usepackage{amsthm}
\usepackage{breakcites}
\usepackage{mathtools}
\usepackage{csquotes}
\usepackage{colortbl}
\usepackage{breakurl}
\usepackage{natbib}
\usepackage{cite}
\usepackage[shortlabels,inline]{enumitem}
\usepackage{arydshln}
\usepackage{mathabx}
\usepackage{appendix}
\makeatletter
\newcommand{\inlineitem}[1][]{
\ifnum\enit@type=\tw@
    {\descriptionlabel{#1}}
  \hspace{\labelsep}
\else
  \ifnum\enit@type=\z@
       \refstepcounter{\@listctr}\fi
    \quad\@itemlabel\hspace{\labelsep}
\fi}
\makeatother
\usepackage{thm-restate}
\usepackage{tabularx}
\usepackage{xcolor}
\usepackage{xspace}
\usepackage{todonotes}
\usepackage{times}
\usepackage{xcolor}
\usepackage{xspace}
\usepackage{soul}

\usepackage{tabularx}
\usepackage{amsmath}
\usepackage{amssymb}
\usepackage{graphicx}
\usepackage{amsthm}
\graphicspath{ / }
\newtheorem{definition}{Definition}
\newtheorem{theorem}{Theorem}
\newtheorem{lemma}{Lemma}

\newcommand{\optgeneral}{\alpha}
\newcommand{\opttruth}{\optgeneral^*}
\newcommand{\optalgo}{\widetilde{\optgeneral}}
\newcommand{\regretalgo}{\beta} 
\newcommand{\defn}{:=}
\newcommand{\setM}{[M]}

\usepackage{soul}
\usepackage[linesnumbered,ruled,norelsize]{algorithm2e}

\SetCommentSty{commfont}

\newcommand{\algotitle}{In\emph{sense}tive
Mechanism\xspace}
\newcommand{\bestutility}{U^*}
\title{An Incentive Mechanism for Crowd Sensing with Colluding Agents}

\author{Susu Xu, Weiguang Mao, Yue Cao, Hae Young Noh, Nihar B. Shah\\Carnegie Mellon University \\ \{susux, weiguanm, yuec1\}@andrew.cmu.edu, noh@cmu.edu, nihars@cs.cmu.edu}
\date{}

\begin{document}

\maketitle

\begin{abstract}
Vehicular mobile crowd sensing is a fast-emerging paradigm to collect data about the environment by mounting sensors on vehicles such as taxis. An important problem in vehicular crowd sensing is to design payment mechanisms to incentivize drivers (agents) to collect data, with the overall goal of obtaining the maximum amount of data (across multiple vehicles) for a given budget. Past works on this problem consider a setting where each agent operates in isolation---an assumption which is frequently violated in practice. 
In this paper, we design an incentive mechanism to incentivize agents who can engage in arbitrary collusions. We then show that in a ``homogeneous'' setting, our mechanism is optimal, and can do as well as any mechanism which knows the agents' preferences a priori. Moreover, if the agents are non-colluding, then our mechanism automatically does as well as any other non-colluding mechanism. We also show that our proposed mechanism has strong (and asymptotically optimal) guarantees for a more general ``heterogeneous'' setting. Experiments based on synthesized data and real-world data reveal gains of over 30\% attained by our mechanism compared to past literature. 
\end{abstract}
\section{Introduction}
Mobile Crowd Sensing (MCS) is a new and fast-rising community sensing paradigm to fulfill the increasing demand of diverse urban sensing data. The recent proliferation of mobile smart devices (e.g., smart phones, smart watches, etc.) provides increasingly capable sources of sensors (e.g., GPS, gyroscope, accelerometer, camera, etc.) to MCS applications. In MCS, crowdsoucer engages a crowd of participants (called agents) to collect sensory data using existing pervasive mobile devices~\citep{ganti2011mobile,lane2010survey}, allowing for collection of diverse data at scale. In an MCS system, after obtaining data from agents, the crowdsoucer cleans and analyzes the data for various applications such as smart city management~\citep{thiagarajan2009vtrack,dutta2009common} and social applications~\citep{eisenman2009bikenet,zheng2013u}. Vehicular MCS is a type of MCS that employs vehicles such as taxis as agents. It is gaining increasing popularity in recent times particularly for applications such as air/noise pollution monitoring~\citep{zheng2013u,devarakonda2013real}, transportation monitoring~\citep{thiagarajan2009vtrack,wan2016mobile},  and infrastructure monitoring~\citep{hull2006cartel}. In vehicular MCS, a crowdsoucer employs monetary incentives to motivate agents (drivers) to share spatial-temporal information collected along with their trajectories. Here, agents first ``bid'' the amount of monetary remunerations they wish to receive from the crowdsourcer for the task of collecting and reporting data. The bid by an agent is supposed to reflect the ``threshold'' of remuneration beyond which they will accept to do the task and below which they will not. The crowdsourcer then makes a monetary offer to each agent, which the agent may accept and perform the task. The overall goal of the crowdsourcer is to get the maximum possible number of tasks done (that is, obtain the maximum amount of data) for a given budget, and the goal of the agents is to maximize the monetary remunerations they receive.

Various past works have studied and developed incentive mechanisms to motivate agents to bid truthfully in crowdsensing applications; see~\citep{duan2012incentive,yang2012crowdsourcing,zhao2014crowdsource,jin2017theseus} and references therein. The methods proposed in past literature, however, operate under the assumption that the agents (drivers) never communicate or collude with each other. This assumption of non-collusion is known to be frequently violated in practice. For example, taxi drivers can discuss and collaborate via their radio systems or in person. Such collusions present significant challenges to the design of crowd sensing systems and the collection of quality data with an efficient utilization of the crowdsourcer's budget. For instance, the recent paper~\citep{ji2017designing} warns that collusions among agents may result in ``serious problems for the platform, such as losing
system-wide utility and deteriorating other users’ enthusiasm
to participate.'' A number of other research studies~\citep{wang2013artsense,wang2014enabling} have shown that with the increasing of colluding agents, it becomes more difficult to eliminate the effects of collusion attacks on the crowd sensing system. A second approach to combat collusions -- at least in crowdsourcing applications -- is to assume knowledge of some ``gold standard'' ground truth~\citep{gneiting2007strictly,lambert2009eliciting,shah2016double,ShahZhouPeres2015,shah2016selfcorrection} and design proper scoring rules. However, this is a very strong assumption in our crowdsensing setting, and we will not make any such assumption in this paper. 

In this paper, we propose a novel incentive mechanism -- termed  ``\algotitle''\footnote{The name captures the fact that it is an \emph{in}cen\emph{tive} mechanism to crowd \emph{sense}, and also highlights its primary feature that it is insensitive to collusions among agents.} -- for the crowd sensing problem with colluding agents. We prove that in a setting where the thresholds of agents are ``homogeneous'', our \algotitle is optimal in that truth-telling is optimal for the agents, and the number of tasks is maximized for the crowdsourcer when agents report truthfully. We then move on to a more general ``heterogeneous'' threshold setting where we first prove the impossibility of any mechanism that can ensure truth-telling; we then show that our proposed \algotitle continues to fare quite well even in this setting. In numerical evaluations, we find that our proposed mechanism consistently outperforms the state-of-the-art and can lead to gains of 30\% or higher in the amount of data collected under a given budget. 

The rest of the paper is organized as follows. We formally introduce the problem in Section~\ref{sec:problem}. We then present our main results including our mechanism, theoretical guarantees, and numerical evaluations, in Section~\ref{sec:incentive}. We conclude the paper with a discussion in Section~\ref{sec:conclu}. This paper also includes proofs of our main theoretical results in appendix.

\section{Problem formulation}
\label{sec:problem}

The system consists of a crowdsourcer with budget $B > 0$, and $M \ge 2$ agents, whom we index as $1,\ldots,M$. Each agent has some private threshold $T_{m}^{*} \geq 0$ such that the agent is incentivized to collect data in that round if and only if the remuneration provided by the crowdsourcer is at least $T_{m}^{*}$. This threshold captures the willingness of the agent to carry out the data collection task, and is assumed to be constant across the $N$ rounds. This assumption is reasonable as the duration of each round and the entire game are usually short (few hours), and consequently the thresholds of drivers vary little across the rounds~\citep{devarakonda2013real,hasenfratz2015deriving}. The threshold is known to that agent, but unknown to the crowdsourcer. The entire process comprises $N \ge 1$ rounds of data collection. At the beginning of each round (say, round $n \in \{1,\ldots,N\}$), the agents may communicate with each other, and then each agent (say agent $m$) bids a number $T_{m, n}$ to the crowdsourcer. 
This bid is supposed to represent the true threshold of that agent; however, the agents may be strategic and report a larger number on purpose.
We allow for every agent to be cognizant of each others' true thresholds as well as bids. The crowdsourcer then offers each agent (say, agent $m$) a reward $R_{m,n}$ if the agent agrees to collect data in the current round $n$. Each agent may then accept or reject the reward. 
If an agent accepts, then the agent collects and reports data to the crowdourcer in that round and obtains the reward $R_{m,n}$; otherwise the agent does not collect data or receive the reward. We term the data obtained by the crowdsourcer from each agent for each round as one data point. 

The overall goal of the crowdsourcer is to maximize the number of data points obtained across $N$ rounds, subject to the budget constraint. As we assume that the entire incentive mechanism employed is publically known, it is possible for the agents to collude and exploit the mechanism to maximize their own utility. We formulate the objective of the crowdsourcer and the agents in the sections below.

\subsection{Objective of the crowdsourcer}
The objective of the crowdsourcer is to design an appropriate incentive mechanism to ensure that after $N$ rounds, the budget is best utilized to collect as much data as possible. The objective function refers to the amount of accepted tasks, which equals to the final amount of collected data. We formulate the optimization problem for crowdsourcer in Equation~\eqref{eq1} below:
\begin{subequations}
\begin{align}
    \max_{\{R_{m,n}\}_{m=1\ldots M,n=1\ldots N}} &~~\sum_{n=1}^N\sum_{m=1}^M I(R_{m,n}\geq T_{m}^{*})\label{eq1a}\\
    \textrm{such that }& \sum_{n=1}^N\sum_{m=1}^M R_{m,n}\leq B\label{eq1c}
\end{align}
\label{eq1}
\end{subequations}
\noindent Here $I$ denotes the indicator function. The objective of the crowdsourcer is to maximize the total number of data points collected~\eqref{eq1a}, or in other words, the total number of tasks accepted by the agents. Note that any agent $m$ during round $n$ will agree to perform the task if and only if $R_{m,n} \geq T^*_{m}$. The threshold $T_{n}^{*}$ represents the minimum monetary incentive required for the agent motivate herself/himself to accept a task~\citep{angelopoulos2014characteristic}. We also note that the crowdsourcer cannot simply solve the problem~\eqref{eq1} as a standard optimization problem because it does not know the true thresholds $\{T^*_{m}\}_{m=1}^{M}$ of the agents.

In this paper, we will measure the efficiency of any mechanism in terms of its \emph{regret} defined as follows. We let $\opttruth$ denote the optimal value of the objective~\eqref{eq1} that the crowdsourcer can achieve in the hypothetical situation that all agents always report their thresholds truthfully. In this case, the values of thresholds $\{T^*_{m}\}_{m=1}^{M}$ become known to the crowdsourcer, and the problem~\eqref{eq1} then turns out to be a simple packing problem. Then for any proposed mechanism to elicit the thresholds as bids from the agents, we denote the (expected) value of the objective~\eqref{eq1a} achieved by the mechanism as $\optalgo$. We then define the regret $\regretalgo$ of this mechanism as
\begin{align*}
    \regretalgo \defn   \opttruth - \optalgo.
\end{align*}
\emph{The goal is to design incentive mechanisms that have the smallest regret $\regretalgo$ possible.} 

\subsection{Objective of agents}
The objective of an agent is to maximize the total remuneration they earn across the $N$ rounds. Specifically, any agent $m$ will report thresholds $\{ T_{m, n} \}_{n=1}^{N}$ that maximize the objective
\begin{align}
    \sum_{n=1}^N R_{m,n} ~I(R_{m,n} \geq T_m^*).
\label{eq2}
\end{align}
Here, $\{ R_{m,n} \}_{n=1}^N$ is the set of remunerations provided by the crowdsourcer's mechanism for the reported bids of the agents.\footnote{One may additionally want to subtract out any costs incurred by the agents from their objective. However, in the practical setting of crowd sensing from taxi drivers, the agents suffer only a negligible cost~\citep{zhang2016incentives} as they only need to perform a few operations on their phone.} 
Finally, if there exist multiple strategies to obtain the maximum profit, the agent will follow the strategy that requires accepting the smallest number of tasks.

\subsection{Equilibrium of colluding agents}
In this section we discuss the idea of equilibrium among agents that might collude in price rigging. We begin by noting that because of the possible collusion, the widely used Nash equilibrium is inappropriate in our context. Hence we instead employ the notion of \textit{perfect cooperative equilibrium} (PCE)~\citep{rong2014cooperative} that is better suited to model colluding agents. 

In order to formally define PCE, we first introduce the concepts of best response and best utility. Consider an $M$-agent game $G$ with a set of agents denoted as $\setM$. Let $S_i$ be the set of strategies for any agent $i \in \setM$, and $s_i \in S_i$ the strategy that agent $i$ will take, $s_{-i}$ the (vector of) strategies that agents in $\setM \backslash \{i\}$ take in game $G_{s_i}$. The utility obtained by agent $i$ taking strategy $s_i$ and others taking $s_{-i}$ is denoted as $U_i(s_i, s_{-i})$.

The notion of \textbf{best response} in the game $G$ is then defined as such: Let $G_{s_i}$ be the ($M$-$1$)-agent game among agents in the set $\setM \backslash \{i\}$ when agent $i$ takes strategy $s_i$. 
The strategy $s_i$ is called a best response in game $G$ if $s_i$ maximizes agent $i$'s expected utility $U_i(s_i, s_{-i})$ 
given that the other agents are playing a Nash Equilibrium $s_{-i} \in$ $NE^G(s_i)$ in $G_{s_i}$. Given the fixed budget $B$, $\bestutility_i$ is the upper bound of utility agent $i$ can obtain for any possible strategy $s_i$ when other agents reach a Nash Equilibrium $NE^G(s_i)$ in $G_{s_i}$ correspondingly. 
We define this upper bound, $\bestutility_i$ as the \textbf{best utility}.
\begin{align*}
    \bestutility_i \defn \sup_{s_i \in S_i, s_{-i} \in NE^G(s_i)} U_i(s_i, s_{-i}).
\end{align*}
With these preliminaries, we are now ready to define PCE.
\begin{definition}[Perfect cooperative equilibrium, PCE]
A strategy profile $s$ is a perfect cooperative equilibrium (PCE) in an $M$-agent game if for every agent $i \in \setM$, we have
\begin{align*}
     U_i(s_i, s_{-i}) \ge \bestutility_i .
\end{align*}
\end{definition}
The definition of PCE coincides with the intuition if a strategy generates greater or equal utility for an agent than the utility obtained by her/him
acting selfishly, then it is stable and will be favored by all agents.

\section{Main results}
\label{sec:incentive}
In this section, we present our proposed \algotitle and provide associated guarantees.

\subsection{\algotitle}
\label{sec:algo}
\begin{algorithm}[ht]
 \SetKwInOut{Input}{Input}
 \SetKwInOut{Output}{Output}
 \DontPrintSemicolon
 \Input{Total budget $B$, Number of rounds $N$, Number of agents $M$, Set of agents $\setM=\{m\}_{m=1}^M$}
 \Output{ incentives collected $\{R_{m,n} | m \in \setM, n \in N \}$}
 initialize $R_{m,n} \gets 0$ for each agent $m$, round $n$\;
 initialize thresholds $T_{m,0} \gets \infty$ for each agent $m$\;
 initialize allocated budget $B_{used} \gets 0$ \;
 broadcast $B, N, M$ \; 
 \For{$n \gets 1$ \KwTo $N$}{
    \For{$m \gets 1$ \KwTo $M$}
    {
    \tcp{Thresholds update}
        collect submitted bid of agent $m$ as $T_{m,n}$\;
        $T_{m,n} \gets min(T_{m,n}, T_{m, n-1})$ \;
    }
    \tcp{Calculate incentive}
    $R \gets \min \Big(T_{1,n}, \cdots, T_{M,n}, \dfrac{B - B_{used}}{N-n+1} \Big)$ \;
    \tcp{Select agents}
    $K \gets \Big\lfloor \dfrac{B - B_{used}}{R \times (N-n+1)} \Big\rfloor$ \;
    
    stable sort the list of agents $\{m\}_{m=1}^M$ in ascending order according to their thresholds 
    $\{T_{m,n}\}_{m=1}^M$, choose top $K$ agents as $selected$\;
    \tcp{Task allocation}
    \For{$m \in selected$}
    {
        assign task to $m$ with incentive $R$\;
        \If{$m$ accepts task} 
        {
            $R_{m,n} \gets R$  \tcp{Assign incentive}
            $T_{m,n} \gets R_{m,n}$ \tcp{Calibrate threshold}
            $B_{used}\gets B_{used} + R_{m,n}$ \;
        }
    }
}
 \caption{\algotitle}
 \label{algorithm1}
\end{algorithm}

We begin with an intuition of our proposed \algotitle. To achieve the objective of the crowdsourcer, we need an incentive mechanism that prevents agents from bid rigging. The key idea here is to differentiate the agents so that for any potential collusion, there exist some disadvantaged agents who are incentivized to betray and compete. Compared to conventional incentive mechanisms~\citep{prelec2004bayesian,kamble2015truth,jin2017theseus,zheng2017budget} which treat players symmetrically, \algotitle~ is asymmetric, in the sense that it selects and pays the agents depending on their (arbitrary yet fixed) indices. As a simple example to illustrate this point, consider two agents with indices $i < j$, who agree to always bid at the same price. If we can only afford to choose one agent for each round, then our \algotitle always selects agent $i$. A second key idea in the construction of \algotitle is to identify and address two extremal strategies that could be excuted by the colluding agents, and suitably smooth out the resulting mechanism to handle all non-extremal strategies as well. The first extremal strategy we identify is when every agent bids very high initially hoping that the mechanism will be fooled into believing that their true thresholds are very high and thereby paying a large amount of budget for fewer tasks. In order to circumvent this attack, our algorithm imposes a carefully chosen upper bound on the reward paid in each round, and this upper bound depends on the budget and the number of rounds that remains. The second extremal strategy we identify is when every agent bids low initially and rejects the rewards from initial rounds so that the mechanism may pay a higher amount (per task) for later tasks due to a large leftover budget. We foil this attack by ensuring that the budget is distributed across different rounds in a relatively uniform manner. 

With these intuitions, we are now ready to formally present the \algotitle in Algorithm~\ref{algorithm1}. We parse through the details of the mechanism in the remainder of this subsection. The mechanism first calibrates current estimation of each agent's threshold based on their bids in the current round and the estimated thresholds from the previous round. After updating the thresholds of all the agents, $R$, the amount of incentive given for completing one task in the current round, is calculated by taking the minimum among $T_{1,n}, \cdots, T_{M,n}$. Notice that $R$ is upper bounded by $(B-B_{used})/(N-n+1)$, which helps even out the amount of budget spent in each round. This also implies that if $    \min(\{T_{m,n}\}_{m=1}^M)>\dfrac{B-B_{used}}{N-n+1}$, then 
we still try to pay one agent $(B-B_{used})/(N-n+1)$ instead of not paying anyone, as a strategy to elicit a potentially lower threshold from some dishonest agents. Based on the amount of reward $R$ for a single task, we decide $K$, the number of agents we select in this round. 

 The selection process is done by stable sorting the list of all agents based on their bids. The original list is ordered by agent's index, so that 1) agents with smaller bids are prioritized over those with larger bids, and 2) agents with smaller indices are always preferred over those with larger indices. This stable sort step is known to all agents, and intentionally creates a differentiation among agents. Thus in any potential collusion, there exists a most disadvantaged agent who is always the last to be selected and can only collect the leftover compared to other group members. Under our \algotitle, this disadvantaged agent knows that she/he can only gain more if 1) she/he is prioritized in the selection step, which means she/he has to bid lower than others, or 2) there is more budget left when she/he gets selected, which means that agents with smaller indices must get fewer payment per round. In any case, this agent is incentivized to defect and bring down the bids.

Finally, if every agent bids higher than the upper bound for each round in order to save budget for later rounds, we will try to pay $R$ to see if some of the agents will accept it. If the task is accepted with $R$, we will update $B_{used}$ and our estimation of agent's threshold $T_{m,n}$.

\subsection{Optimality with homogeneous thresholds}

In this section, we present the optimality guarantee for our proposed \algotitle for a setting with ``homogeneous'' thresholds. In this homogeneous-threshold setting, we assume that all agents have the same threshold, that is, $T_{1}^{*}=\cdots=T_{M}^{*}=T^*$. The following results show that our mechanism is optimal in this homogeneous setting.

\begin{theorem}
In the homogeneous threshold setting, under the \algotitle, the honest strategy is a perfect cooperative equilibrium (PCE) for all agents, and the \algotitle achieves a zero regret.
\label{theorem1}
\end{theorem}

This result shows that when the thresholds are homogeneous, the \algotitle can surprisingly achieve the same performance as if the thresholds were known exactly to the algorithm. Consequently, even if the agents did not collude, the \algotitle performs at least as well as any algorithm designed specifically to exploit the non-colluding assumptions. The proof of Theorem~\ref{theorem1} is provided in Appendix~\ref{appendix:theo1}.

\subsection{Bounded regret with heterogeneous thresholds}
\label{sec:hetero}

In this section, we consider a more general situation where different agents can have different thresholds. We present both negative and positive results for this setting. 

\subsubsection{Suboptimality of any mechanism}
Recall that in the homogeneous threshold setting, our \algotitle achieves a zero regret. Given this result, a natural question that arises is whether there exists an algorithm that can provide a zero regret guarantee for the heterogeneous setting as well. In this section we show that unfortunately no mechanism can provide such a strong guarantee in this general setting.
\begin{theorem}
In the heterogeneous threshold setting, no mechanism can ensure that honest strategy is a perfect cooperative equilibrium (PCE).

\label{theorem3}
\end{theorem}
The proof of this result is provided in Appendix~\ref{proof3.2}.

\subsubsection{\algotitle~in heterogeneous situation}\label{sec:application}
We now provide an upper bound on the regret of our \algotitle in the general heterogeneous threshold setting.  

\begin{theorem}
The regret $\regretalgo$ attained by \algotitle in the heterogeneous setting is upper bounded as
\begin{align*}
    \regretalgo \leq \nonumber
    &\begin{cases}
    \opttruth -iN &\emph{if } iT_i^{*}< B/N\leq iT_{i+1}^{*}\\
    \opttruth - \Big(iN + \Big\lfloor \frac{B - iNT_{i+1}^{*}}{T_{i+1}^{*}}\Big \rfloor\Big) &\emph{if }iT_{i+1}^{*}< B/N\leq (i+1)T_{i+1}^{*}\\
    \opttruth - \Big\lceil \frac{B}{T_2^{*}}\Big \rceil&\emph{if } B/N \leq T_{2}^{*}\\
    \opttruth- MN&\emph{if } B/N> MT_{M}^{*},
    \end{cases}
\end{align*}
where
\begin{align*}
\opttruth=\sum_{i=0}^M \min \Big(\Big\lfloor \frac{B - N\sum_{k=1}^iT_k^{\star}}{T_{i+1}^{\star}}\Big \rfloor,N \Big).
\end{align*}
Moreover, for every value of $N$, $M$ and heterogeneous thresholds, the algorithm attains $\regretalgo=0$ when $B$ is large enough.
\label{prop1}
\end{theorem}
The proof of this result is presented in Appendix~\ref{proofappendix}.

\subsubsection{Numerical experiments} 
\begin{figure}[h]
  \centering
    \includegraphics[width=0.5\textwidth]{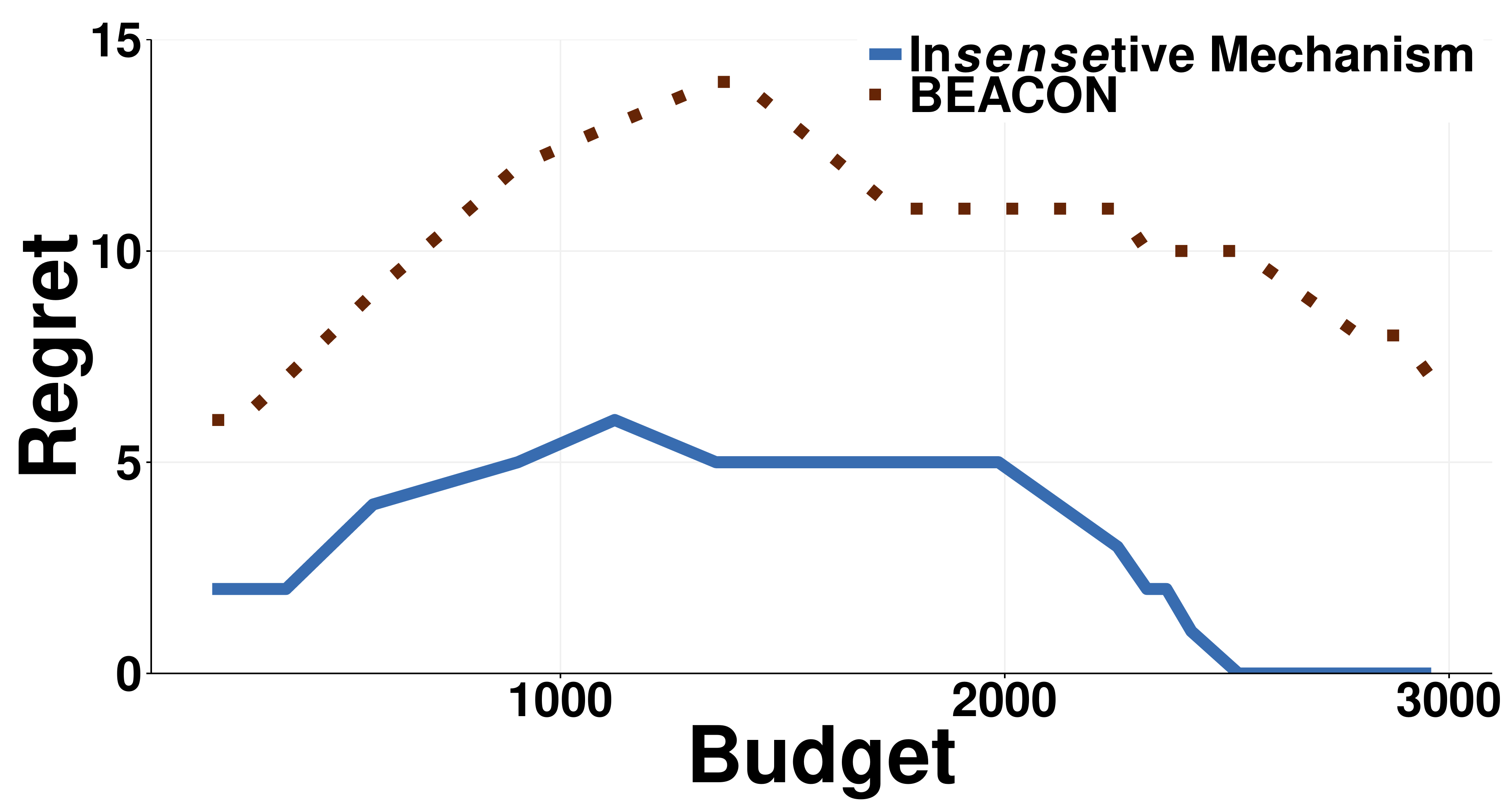}
    \caption{A heterogeneous threshold scenario where the game has $5$ rounds and $5$ vehicles with real thresholds of $20,40,50,70$ and $100$. The regret $\regretalgo$ of our \algotitle~(solid blue line) is consistently and significantly smaller as compared to the prior state-of-the-art  \emph{BEACON} (dotted brown line). 
    \label{fig1}}
\end{figure}

We now present some numerical simulations that compare the performance of our \algotitle with the state of the art. As discussed earlier, past literature focuses on a non-colluding setting, and here we compare with \emph{BEACON}~\citep{zheng2017budget} which is the well-known mechanism for the non-colluding setting. 

Figure~\ref{fig1} compares the regret $\regretalgo$ of our \algotitle with \emph{BEACON} for a 5-round 5-agent game in the heterogeneous setting where agents have true thresholds as $20,40,50,70,100$. The budget ranges from $20$ to $3,000$. Overall, our mechanism collects $82.94\%$ of the total number of data points compared to $\opttruth$, whereas {BEACON} attains only $46.63\%$ of $\opttruth$ on average. Compared to \emph{BEACON}, our \algotitle achieves an improvement of $36.31\%$. 

\subsubsection{Experiments on real-world data} 
\begin{figure}[h]
  \centering
    \includegraphics[width=0.5\textwidth]{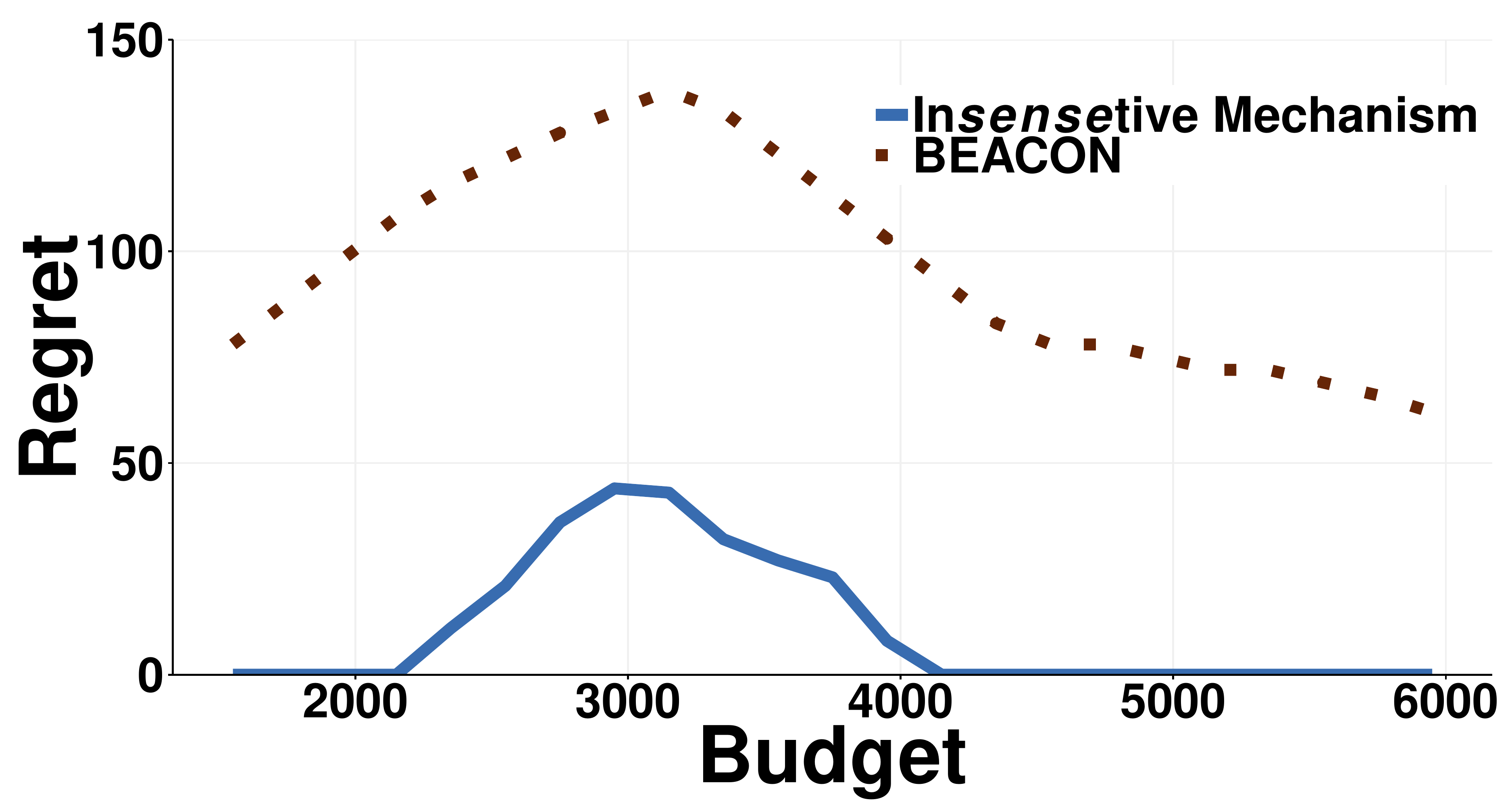}
    \caption{A real-world scenario where the game has $6$ rounds and $50$ vehicles with real heterogeneous thresholds. With different amounts of budget, the regret $\regretalgo$ of our \algotitle~(solid blue line) is significantly smaller than the baseline method \emph{BEACON} (dotted brown line). 
    \label{fig2}}
\end{figure}

We collected the real-time surge rates, which reflect how many times of the base fare drivers expect to get after riding passengers, from the real-world ride-sharing companies. We used these heterogeneous surge rates to represent the real heterogeneous thresholds of different cars. Therefore, the real threshold of each car was estimated by combining the base fare during a fixed length of time, which was a constant for all cars, and its individual surge rate. We collected the data of around 5,000 cars during 11 days. The data showed that the surge rates of most cars slightly vary with time and location.

Figure~\ref{fig2} compares the regret $\regretalgo$ of our \algotitle with \emph{BEACON} for a 6-round 50-agent game with different budgets. The budget ranges from 1,500 to 6,000. Compared to $\opttruth$, our \algotitle manages to collect $96.21\%$ of the total data points , while {BEACON} only acquires $63.08\%$ of $\opttruth$ on average. Our mechanism improves the performance of $33.13\%$ compared to \emph{BEACON}. 
\begin{figure}[h]
  \centering
    \includegraphics[width=0.5\textwidth]{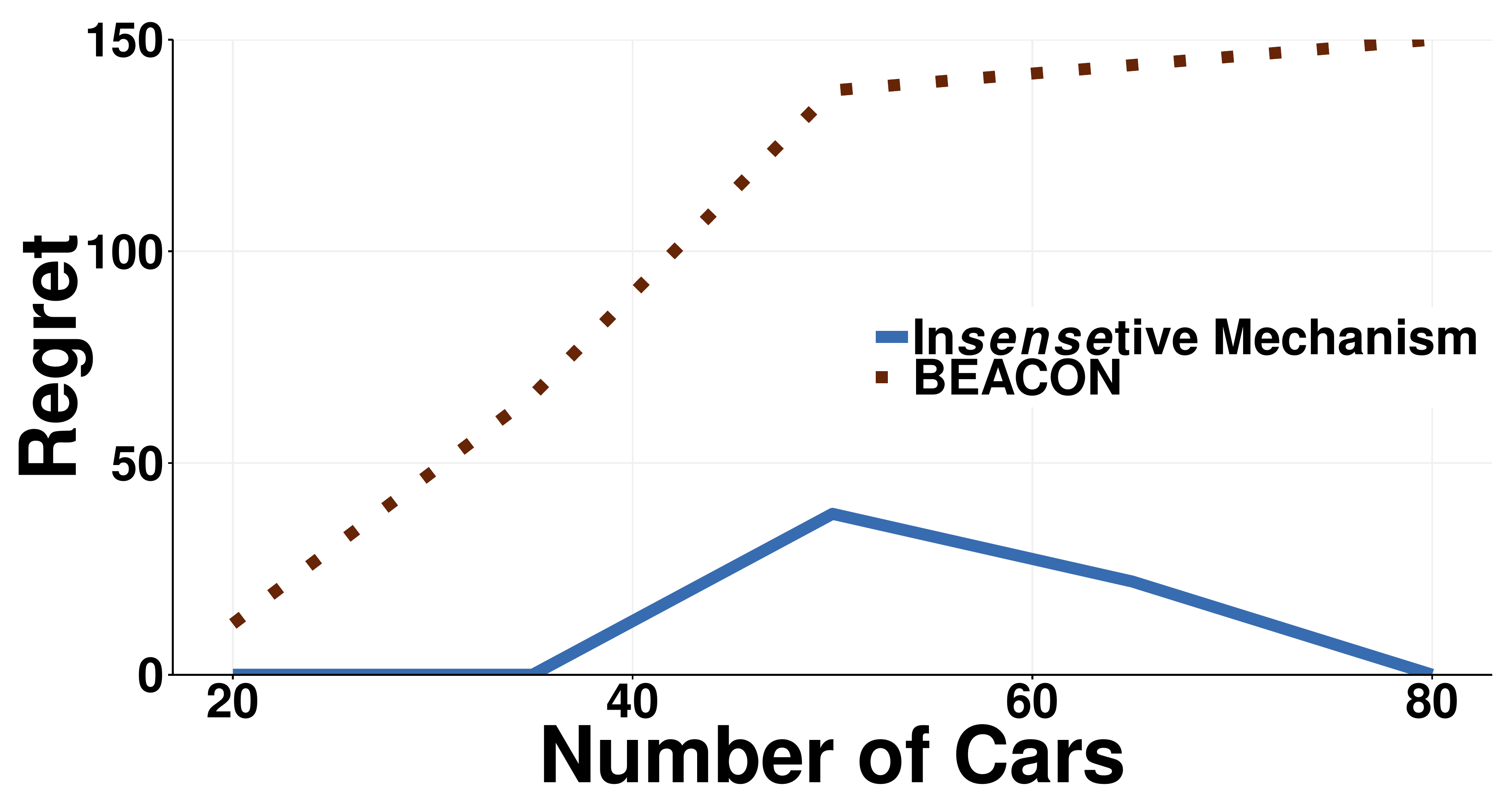}
    \caption{A real-world scenario where the game has $6$ rounds and a fixed budget of $3,000$. As the number of vehicles varies, the regret $\regretalgo$ of our \algotitle~(solid blue line) consistently outperforms the baseline method \emph{BEACON} (dotted brown line). 
    \label{fig3}}
\end{figure}

Figure~\ref{fig3} compares the regret $\regretalgo$ of our \algotitle with \emph{BEACON} for a 6-round game with a fixed budget of $3,000$ and the number of cars ranges from $20$ to $80$. As the number of cars increases, the regret of \emph{BEACON} increases accordingly. However, the regret of our \algotitle first achieves the peak with $50$ cars and then decreases, which demonstrates the robustness of our mechanisms. Given a fixed budget, our mechanism can collect $95.86\%$ of the total number of points, which achieves an improvement of $33.52\%$ on average compared to \emph{BEACON}.

\section{Discussion}
\label{sec:conclu}
Vehicular crowd sensing is a fast emerging paradigm to collect data for various modern applications. In this paper, we design the \algotitle to incentivize agents to collect and report data, with a focus on addressing the challenge of collusion among agents. We show that if the setting is homogeneous in nature, our mechanism performs as well as the best algorithm when the agents' preferences are known (that is, achieves a zero regret). We also provide negative and positive results for the heterogeneous setting: no mechanism can ensure a zero regret, while our mechanism performs reasonably well. 

A key idea behind our \algotitle that helps to combat collusion is the introduction of a bias in the choice of the agents. Interestingly, while biases are often undesirable in applications such as statistical learning~\citep{nie2017adaptively,wang2018your}, we find that a deliberate introduction of bias is beneficial in this game-theoretic setting.

In future work, we wish to  incorporate the learning of the spatial-temporal distribution of thresholds of agents from their history. The goal is to design incentive mechanisms that can exploit this ``prior'' knowledge in order to achieve a lower regret in the heterogeneous setting.

\medskip
\noindent{\textbf{Acknowledgments.}} This work was supported in part by Carnegie Mellon University’s Mobility21 National University Transportation Center (grant number 69A3551747111), which is sponsored by the US Department of Transportation, the Dowd Fellowship from the College of Engineering at Carnegie Mellon University, Intel, Google, and NSF grants CRII: CIF: 1755656 and CCF: 1763734.

\appendices

\section{Proof of Theorem~\ref{theorem1}}
\label{appendix:theo1}

In our proposed mechanism, given $B \leq MNT^{*}$, the honest strategy produces a deterministic reward for each agent:
\begin{subequations}
\begin{align}
    U_1(s) &= \min\Big(N, \Big\lfloor \frac{B}{T^{*}} \Big\rfloor \Big)T^{*} \label{eq8a}\\
    U_2(s) &= \min \Big(N, \Big\lfloor \frac{B - U_1}{T^{*}}\Big \rfloor \Big) T^{*}\label{eq8b}\\
    &\ldots\nonumber \\
    U_M(s) &=  \Big\lfloor \frac{B - U_1(s)- \ldots - U_{M-1}(s)}{T^{*}}\Big \rfloor T^{*}\label{eq8c}
\end{align}
\label{eq:finalresult}
\end{subequations}

The Equation~\eqref{eq8a},~\eqref{eq8b} and~\eqref{eq8c} naturally follow the mechanism that when all agents bid the same, agents with smaller indices always have a higher priority.

\begin{lemma}
In our incentive mechanism, for any round $n\in \{1, 2, 3, \ldots, N-1\}$, the number of selected agents in the round $n$ is not larger than the number of selected agents in the round $n+1$.  
\label{lemma1}
\end{lemma}
\begin{proof}
If this theorem does not hold, then there must exists a $n' \in \{1, 2, \ldots, N\}$, where the number of selected agents at round $n'$, $K_{n'}$, is larger than the number of selected agents at round $n+1$, denoted as $K_{n'+1}$. That means there must be at least one agent who is selected in round $n'$ but not selected in round $n'+1$. Because in each round $n'$ we select $K_{n'} = \frac{B - B_{used}}{N - n' + 1}$ agents. If we denote the budget left in round $n'$ as $B_{n'}$, and the reward for each selected agents in round $n'$ and $n'+1$ as $T_{n'}, T_{n'+1}$, this situation can be described as 
\begin{subequations}
\begin{align*}
    \dfrac{B_{n'}}{N- n' + 1}\ge K_{n'}T_{n'}, &\dfrac{B_{n'} - K_{n'}T_{n'}}{N - (n' + 1) + 1} < K_{n'+1}T_{n'+1}\\
    K_{n'}>K_{n'+1}, &T_{n'}\geq T_{n'+1}
\end{align*}
\end{subequations}
Expand the inequalities we have 

\begin{equation*}
    B_{n'} \ge K_{n'}(N- n' + 1)T_{n'}, B_{n'}  < K_{n'}(N- n' + 1)T_{n'}
\end{equation*}
which leads to a contradiction. Therefore the lemma holds.
\end{proof}
With this lemma, we prove that the honest strategy $s$ is a PCE by splitting the problem into 3 cases.

\textbf{Case 1.} $B \le NT^{*}$: 
For the first agent, if she/he bids honestly throughout the game, then the best response of other agents will have no effect on the her/his total reward $ \lfloor B / T^{*} \rfloor T^{*}$. Otherwise, if she/he bids $T_{1,1}> T^{*}$ in the first round, then the Nash Equilibrium for the remaining $m-1$ agents is that all agents bid honestly as $T^{*}$.  
If the first agent bids $T^{*}$, but increase her/his bid in the round $n>1$, then the mechanism will use her/his previous bid $T^{*}$ as her/his reward and she/he would accept the reward to finish the task according to our assumption. In this case, the total reward is still $ \lfloor B / T^{*} \rfloor T^{*}$. Therefore, we have 
 $   \bestutility_1 = \max(0, \lfloor B / T^{*} \rfloor T^{*}) =\lfloor B / T^{*} \rfloor T^{*}$
. The reward of agent 1 is upper bounded by the reward gained adopting the honest strategy $U_1(s)=\lfloor B / T^{*} \rfloor T^{*}$. Therefore, we have 
$U_1(s) = \bestutility_1$. For the other agents $i\in \setM\backslash\{1\}$ and any strategy $s_2$, the Nash Equilibrium is still that all these agents bid honestly as $T^{*}$. In this case, the first agent receives $ \lfloor B / T^{*} \rfloor T^{*}$ and other agents receive $0$. As $B \le NT^{*}$, we know that $U_i(s) = \bestutility_i = 0$. Since for any $i$, we have $U_i(s) \ge \bestutility_i$, the honest strategy $s$ is a PCE when $B \le NT^{*}$.

\textbf{Case 2.} $B \ge MNT^{*} $:
In this case, the "truth-telling" strategy for crowdsoucer is that at least one agent bids in the range of $T\in [T^{*}, B/MN]$. It is obvious that for any agent $i$, her/his total reward is maximized when she/he bids $B /M T^{*}$ in each round, and collecting a reward of $B/M$. If she/he ever attempts to bid higher, the Nash Equilibrium of the other agents will make her/him no gain in this round, and thus in the whole game. Adopting the honest strategy $s$ where each agent bids $B /M T^{*}$, we have 
    $U_i(s) = \bestutility_i = B/M$
for any $i$. Hence, $s$ is a PCE when $B \ge MNT^{*}$.

\textbf{Case 3.} $iNT^{*} < B < (i+1)NT^{*}~~\forall i \in \{1,\cdots, M-1\}$:
This is the scenario where the budget is enough to select $i$ agents during all $N$ rounds, but not enough to pick up $i+1$ agents during all $N$ rounds.  According to Eq.~\ref{eq:finalresult}, there is 
\begin{subequations}
\begin{align*}
    U_1(s) &= \cdots=U_i(s)=NT^{*}\\
    U_{i+1}(s) &=  \Big\lfloor \frac{B - \sum_{j=1}^i U_j(s)}{T^{*}}\Big \rfloor T^{*}=\Big\lfloor \frac{B - iNT^{*}}{T^{*}}\Big \rfloor T^{*}\\
    U_{i+2}(s)&=\cdots= U_M(s)=0
\end{align*}
\end{subequations}

According to the Lemma ~\ref{lemma1}, there exists a transition round $1\leq n\leq N$ and we also know that the $i+1$ agent is selected from the $N-\lfloor \frac{B - iNT^{*}}{T^{*}}\rfloor+1$ round.

For any agent $j<i+1$, if she/he sticks to the honest strategy $s_j = s$, then $U_j(s_j) = NT^{*}$ no matter what other agents bid. If any of agent $j<i+1$ chooses to give a higher bid, we denote $n$ as the first round where she/he does so. If $n > 1$ then the outcome does not change. But if $n = 1$, i.e., she/he bids $T_{j,1} > T^{*}$ from the first round, then the agent $i+1$'s best response is to bid $T_{i+1,1}=T^{*}$. 
When agent $i+1$ takes the best response to bid $T^{*}$, according to the Lemma~\ref{lemma1}, she/he will receive an extra reward of $T^{*}$. 
Therefore the upper bound of the reward for agent $j$ is $NT^{*}$. Hence $U_j(s)= \bestutility_j=NT^{*}$. 

For agent $i+1$, similarly, if she/he bids higher than $T^{*}$ from the round $N-\lfloor \frac{B - iNT^{*}}{T^{*}} \rfloor+1$, then the best response of any agent $j>i+1$, is to bid $T^{*}$ to get extra reward. Therefore we have $U_i(s)=\bestutility_i=\lfloor \frac{B - iNT^{*}}{T^{*}} \rfloor T^{*}$.

For any agent $j>i+1$, 
the best strategy of any agent $j'\leq i+1$ is to bid honestly as $T^{*}$. Therefore, $\bestutility_j = 0$. Then we have
\begin{align*}
\begin{cases}
     U_j(s)  \ge \bestutility_j = NT^{*}~~~&\text{if }j <i+1\\
      U_j(s) \ge \bestutility_j = \Big\lfloor \frac{B - iNT^{*}}{T^{*}}\Big \rfloor T^{*}~~~&\text{if }j = i+1\\
    U_j(s) \ge \bestutility_j = 0   ~~~&\text{if }j > i+1
\end{cases}
\end{align*}
In summary, the honest strategy is a PCE in all three cases.

If the real threshold $T^{*}$ satisfies $T^{*} \ge B / NM$, at most $\lfloor B/T^{*} \rfloor$ agents can be actuated in the game. If each agent bids $T^{*}$ honestly, based on the mechanism, the payment for each agent in each round will stay at $T^{*}$. By the end of the game, according to the mechanism all budget would have been spent to actuate as many agents as possible. Therefore the objective for the crowdsourcer is achieved. If the real threshold $T^{*} < B / NM$, the budget is large enough to pay each agent in every round, as long as all agents bid the same value between $T^{*}$ and $B/NM$ inclusively.  
The number of agents selected in one round, denoted as $K$, is equal to $M$, and the total number of agents actuated is $NM$.

\section{Proof of Theorem~\ref{theorem3}}
\label{proof3.2}

Given two selected agents $i,j \in \setM$ and an $1$-round game, without the loss of generality, we assume $T_i^{*}<T_j^{*}$ holds. Suppose that there exists an optimal incentive mechanism: the mechanism should optimize the utilization of the budget to ensure $\regretalgo=0$ in any setup of $B$ and real thresholds. that is, the mechanism should always pay real threshold to each agent when budget is not enough to support all agents in all rounds.

We divide all incentive mechanisms into two categories: 1) fair incentive mechanism and 2) unfair incentive mechanism. A mechanism is fair if any agents $i,j\in \setM$ bid the same $T_i = T_j$, agents $i$ and $j$ will be selected with equal probability, which doesn't hold for the unfair case.

If the above optimal mechanism is fair, then the utility of agent $i$ by taking strategy $s'$ of bidding as $T_j^{*}$ will be
\begin{align*}
    U_i(s') = U_j(s) = T_j^{*} > T_i^{*}.
\end{align*}
If the above optimal mechanism is unfair, since the real threshold is unknown, without the loss of generality, if agent $i$ has a higher priority than agent $j$, then the utility of agent $i$ by bidding $T_j$ is 
\begin{align*}
    U_i(s') = T_j^{*} > T_i^{*}.
\end{align*}
However, according to PCE condition, we have
$U_i(s)=T_i^{*} \geq U_i(s') = T_j^{*}$
which leads to a contradiction.

\section{Proof of Theorem~\ref{prop1} }\label{proofappendix}

\textbf{Case 1.} $B/N\leq T_2^{*}$:
In this case, since the budget can afford at most $2$ agents. The budget can afford agent $1$ for all rounds, or agent $2$ for several rounds. We assume $T_1^{*}<T_2^{*}$ and agent $1$ comes with a higher priority $P_1>P_2$. Since $B/N\leq T_2^{*}$, there must exist a strategy for agent $1$ to obtain full reward $B$, which is also the best strategy for agent $1$.
\begin{lemma}
\label{theorem_appendix_1}
When $B/T_2^{*}> \lfloor B/T_2^{*} \rfloor$, the agent $1$ needs to take at least $ \lceil B/T_2^{*} \rceil$ tasks to get the full reward $B$.
\end{lemma}

\begin{proof}
Assume agent $1$ can get reward $B$ by taking number of tasks $k_1<\lceil B/T_2^{*} \rceil$, equivalently $k_1 \leq \lfloor B/T_2^{*} \rfloor$, the average reward for each round would be $B/k_1\geq B/\lfloor B/T_2^{*} \rfloor>T_2^{*}$ according to $B/T_2^{*}> \lfloor B/T_2^{*} \rfloor$. This means that the average reward exceeds $T_2^{*}$, then there must exist some rounds that agent $2$ can bid $T_2^{*}$ to get some reward such that agent $1$ can not get the full reward $B$, which leads to a contradiction.
\end{proof}
When $B/T_2^{*}= \lfloor B/T_2^{*}\rfloor= \lceil B/T_2^{*} \rceil$, the best strategy for agent $1$ would be to bid $T_2^{*}$ every time to get full reward $B$, which needs to take $\lceil B/T_2^{*}\rceil$ tasks. Therefore, given the Lemma ~\ref{theorem_appendix_1} and above statement, we know that for the upper bound for regret is $\opttruth- \lceil B/T_2^{*} \rceil$.

\textbf{Case 2.} $iT_{i+1}^{*} <B/N\leq (i+1)T_{i+1}^{*}$:
In this case, budget $B$ can support the first $i$ agents for $N$ rounds with bidding $T_{i+1}^{*}$ but not enough to support the first $i+1$ agents. For the first $i$ agents, this scenario is the same as the homogeneous threshold case 3:  $iNT^{*}<B<(i+1)NT^{*}$. 
Therefore, at least $iN$ accepted tasks can be obtained. Meanwhile, since we have $B/N>iT_{i+1}^{*}$, there are budget left for the agent $i+1$ in some rounds.
Based on our mechanism, the best strategy for agent $i+1$ would be to bid $T_{i+1}^{*}$. 
Therefore, in this cases, we can finally achieve the upper bound for regret $\opttruth-iN - \Big\lfloor \frac{B - iNT_{i+1}^{*}}{T_{i+1}^{*}}\Big \rfloor$.

\textbf{Case 3.} $iT_i^{*}< B/N\leq iT_{i+1}^{*}$: In this case, the budget can support the first $i$ agents for $N$ rounds with bids $T_{i}^{*}$, but not enough to support all agents for bidding $T_{i+1}^{*}$. Therefore, in this case, there is no chance for agent $i+1$ to get any reward if the first $i$ agents achieve a PCE. 
This case is equivalent to the homogeneous threshold case 2: $B>iNT^{*}$ where $T=T_i^{*}$. Based on the proof in Appendix~\ref{appendix:theo1}, the PCE would be that all of first $i$ agents bid $B/iN$. Thus, each one will obtain the reward of $B/i$ and finish $N$ tasks. The upper bound for regret would be $\opttruth-iN$.

\textbf{Case 4.} $B/N>MT_M^{*}$:
This case is the same as the homogeneous threshold case: $B>MNT^{*}$ where $T^{*}=T_M^{*}$. Based on the proof shown in Appendix~\ref{appendix:theo1}, the upper bound for regret would be $\opttruth-MN$.

\bibliographystyle{apalike}
\bibliography{reference}

\begin{thebibliography}{}

\bibitem[Angelopoulos et~al., 2014]{angelopoulos2014characteristic}
Angelopoulos, C.~M., Nikoletseas, S., Raptis, T.~P., and Rolim, J.~D. (2014).
\newblock Characteristic utilities, join policies and efficient incentives in
  mobile crowdsensing systems.
\newblock In {\em Wireless Days (WD)}, pages 1--6. IEEE.

\bibitem[Devarakonda et~al., 2013]{devarakonda2013real}
Devarakonda, S., Sevusu, P., Liu, H., Liu, R., Iftode, L., and Nath, B. (2013).
\newblock Real-time air quality monitoring through mobile sensing in
  metropolitan areas.
\newblock In {\em Proceedings of the 2nd ACM SIGKDD international workshop on
  urban computing}, page~15. ACM.

\bibitem[Duan et~al., 2012]{duan2012incentive}
Duan, L., Kubo, T., Sugiyama, K., Huang, J., Hasegawa, T., and Walrand, J.
  (2012).
\newblock Incentive mechanisms for smartphone collaboration in data acquisition
  and distributed computing.
\newblock In {\em INFOCOM}, pages 1701--1709. IEEE.

\bibitem[Dutta et~al., 2009]{dutta2009common}
Dutta, P., Aoki, P.~M., Kumar, N., Mainwaring, A., Myers, C., Willett, W., and
  Woodruff, A. (2009).
\newblock Common sense: participatory urban sensing using a network of handheld
  air quality monitors.
\newblock In {\em Proceedings of the 7th ACM conference on embedded networked
  sensor systems}, pages 349--350. ACM.

\bibitem[Eisenman et~al., 2009]{eisenman2009bikenet}
Eisenman, S.~B., Miluzzo, E., Lane, N.~D., Peterson, R.~A., Ahn, G.-S., and
  Campbell, A.~T. (2009).
\newblock Bikenet: A mobile sensing system for cyclist experience mapping.
\newblock {\em ACM Transactions on Sensor Networks (TOSN)}, 6(1):6.

\bibitem[Ganti et~al., 2011]{ganti2011mobile}
Ganti, R.~K., Ye, F., and Lei, H. (2011).
\newblock Mobile crowdsensing: current state and future challenges.
\newblock {\em IEEE Communications Magazine}, 49(11).

\bibitem[Gneiting and Raftery, 2007]{gneiting2007strictly}
Gneiting, T. and Raftery, A.~E. (2007).
\newblock Strictly proper scoring rules, prediction, and estimation.
\newblock {\em Journal of the American Statistical Association},
  102(477):359--378.

\bibitem[Hasenfratz et~al., 2015]{hasenfratz2015deriving}
Hasenfratz, D., Saukh, O., Walser, C., Hueglin, C., Fierz, M., Arn, T., Beutel,
  J., and Thiele, L. (2015).
\newblock Deriving high-resolution urban air pollution maps using mobile sensor
  nodes.
\newblock {\em Pervasive and Mobile Computing}, 16:268--285.

\bibitem[Hull et~al., 2006]{hull2006cartel}
Hull, B., Bychkovsky, V., Zhang, Y., Chen, K., Goraczko, M., Miu, A., Shih, E.,
  Balakrishnan, H., and Madden, S. (2006).
\newblock Cartel: a distributed mobile sensor computing system.
\newblock In {\em Proceedings of the 4th international conference on Embedded
  networked sensor systems}, pages 125--138. ACM.

\bibitem[Ji and Chen, 2017]{ji2017designing}
Ji, S. and Chen, T. (2017).
\newblock On designing collusion-resistant incentive mechanisms for mobile
  crowdsensing systems.
\newblock In {\em Trustcom/BigDataSE/ICESS, 2017 IEEE}, pages 162--169. IEEE.

\bibitem[Jin et~al., 2017]{jin2017theseus}
Jin, H., Su, L., and Nahrstedt, K. (2017).
\newblock Theseus: Incentivizing truth discovery in mobile crowd sensing
  systems.
\newblock {\em arXiv preprint arXiv:1705.04387}.

\bibitem[Kamble et~al., 2015]{kamble2015truth}
Kamble, V., Marn, D., Shah, N., Parekh, A., and Ramachandran, K. (2015).
\newblock Truth serums for massively crowdsourced evaluation tasks.
\newblock {\em arXiv preprint arXiv:1507.07045}.

\bibitem[Lambert and Shoham, 2009]{lambert2009eliciting}
Lambert, N. and Shoham, Y. (2009).
\newblock Eliciting truthful answers to multiple-choice questions.
\newblock In {\em Proceedings of the 10th ACM conference on Electronic
  commerce}, pages 109--118. ACM.

\bibitem[Lane et~al., 2010]{lane2010survey}
Lane, N.~D., Miluzzo, E., Lu, H., Peebles, D., Choudhury, T., and Campbell,
  A.~T. (2010).
\newblock A survey of mobile phone sensing.
\newblock {\em IEEE Communications magazine}, 48(9).

\bibitem[Nie et~al., 2017]{nie2017adaptively}
Nie, X., Tian, X., Taylor, J., and Zou, J. (2017).
\newblock Why adaptively collected data have negative bias and how to correct
  for it.
\newblock {\em arXiv preprint arXiv:1708.01977}.

\bibitem[Prelec, 2004]{prelec2004bayesian}
Prelec, D. (2004).
\newblock A bayesian truth serum for subjective data.
\newblock {\em science}, 306(5695):462--466.

\bibitem[Rong and Halpern, 2014]{rong2014cooperative}
Rong, N. and Halpern, J.~Y. (2014).
\newblock Cooperative equilibrium: A solution predicting cooperative play.
\newblock {\em arXiv preprint arXiv:1412.6722}.

\bibitem[Shah and Zhou, 2016a]{shah2016double}
Shah, N.~B. and Zhou, D. (2016a).
\newblock Double or nothing: Multiplicative incentive mechanisms for
  crowdsourcing.
\newblock {\em Journal of Machine Learning Research}, 17:1--52.

\bibitem[Shah and Zhou, 2016b]{shah2016selfcorrection}
Shah, N.~B. and Zhou, D. (2016b).
\newblock No oops, you won’t do it again: mechanisms for self-correction in
  crowdsourcing.
\newblock In {\em International Conference on Machine Learning}.

\bibitem[Shah et~al., 2015]{ShahZhouPeres2015}
Shah, N.~B., Zhou, D., and Peres, Y. (2015).
\newblock Approval voting and incentives in crowdsourcing.
\newblock In {\em International Conference on Machine Learning}.

\bibitem[Thiagarajan et~al., 2009]{thiagarajan2009vtrack}
Thiagarajan, A., Ravindranath, L., LaCurts, K., Madden, S., Balakrishnan, H.,
  Toledo, S., and Eriksson, J. (2009).
\newblock Vtrack: accurate, energy-aware road traffic delay estimation using
  mobile phones.
\newblock In {\em Proceedings of the 7th ACM Conference on Embedded Networked
  Sensor Systems}, pages 85--98. ACM.

\bibitem[Wan et~al., 2016]{wan2016mobile}
Wan, J., Liu, J., Shao, Z., Vasilakos, A.~V., Imran, M., and Zhou, K. (2016).
\newblock Mobile crowd sensing for traffic prediction in internet of vehicles.
\newblock {\em Sensors}, 16(1):88.

\bibitem[Wang and Shah, 2018]{wang2018your}
Wang, J. and Shah, N.~B. (2018).
\newblock Your 2 is my 1, your 3 is my 9: Handling arbitrary miscalibrations in
  ratings.
\newblock {\em arXiv preprint arXiv:1806.05085}.

\bibitem[Wang et~al., 2013]{wang2013artsense}
Wang, X.~O., Cheng, W., Mohapatra, P., and Abdelzaher, T. (2013).
\newblock Artsense: Anonymous reputation and trust in participatory sensing.
\newblock In {\em INFOCOM, 2013 Proceedings IEEE}, pages 2517--2525. IEEE.

\bibitem[Wang et~al., 2014]{wang2014enabling}
Wang, X.~O., Cheng, W., Mohapatra, P., and Abdelzaher, T. (2014).
\newblock Enabling reputation and trust in privacy-preserving mobile sensing.
\newblock {\em IEEE Transactions on Mobile Computing}, 13(12):2777--2790.

\bibitem[Yang et~al., 2012]{yang2012crowdsourcing}
Yang, D., Xue, G., Fang, X., and Tang, J. (2012).
\newblock Crowdsourcing to smartphones: Incentive mechanism design for mobile
  phone sensing.
\newblock In {\em Proceedings of the 18th annual international conference on
  Mobile computing and networking}, pages 173--184. ACM.

\bibitem[Zhang et~al., 2016]{zhang2016incentives}
Zhang, X., Yang, Z., Sun, W., Liu, Y., Tang, S., Xing, K., and Mao, X. (2016).
\newblock Incentives for mobile crowd sensing: A survey.
\newblock {\em IEEE Communications Surveys \& Tutorials}, 18(1):54--67.

\bibitem[Zhao et~al., 2014]{zhao2014crowdsource}
Zhao, D., Li, X.-Y., and Ma, H. (2014).
\newblock How to crowdsource tasks truthfully without sacrificing utility:
  Online incentive mechanisms with budget constraint.
\newblock In {\em INFOCOM}, pages 1213--1221. IEEE.

\bibitem[Zheng et~al., 2013]{zheng2013u}
Zheng, Y., Liu, F., and Hsieh, H.-P. (2013).
\newblock U-air: When urban air quality inference meets big data.
\newblock In {\em Proceedings of the 19th ACM SIGKDD international conference
  on Knowledge discovery and data mining}, pages 1436--1444. ACM.

\bibitem[Zheng et~al., 2017]{zheng2017budget}
Zheng, Z., Wu, F., Gao, X., Zhu, H., Tang, S., and Chen, G. (2017).
\newblock A budget feasible incentive mechanism for weighted coverage
  maximization in mobile crowdsensing.
\newblock {\em IEEE Transactions on Mobile Computing}, 16(9):2392--2407.

\end{thebibliography}
\end{document}